\def\year{2022}\relax
\documentclass[letterpaper]{article} 
\usepackage{aaai22}  
\usepackage{times}  
\usepackage{helvet}  
\usepackage{courier}  
\usepackage[hyphens]{url}  
\usepackage{graphicx} 
\usepackage{natbib}  
\usepackage{caption} 
\DeclareCaptionStyle{ruled}{labelfont=normalfont,labelsep=colon,strut=off} 
\usepackage{algorithm}
\usepackage{algorithmic}
\usepackage{comment}

\usepackage{amsthm,amsfonts}
\usepackage{color}  
\newtheorem{lemma}{Lemma}
\newtheorem{theorem}{Theorem}

\usepackage{mathtools,dsfont}

\def \E{{ \mathbb{E} }} 
\def \P{{ \mathbb{P} }}

\DeclarePairedDelimiter\floor{\lfloor}{\rfloor}
\newcommand{\ml}[1]{\log \Big({#1} \Big)}
\newcommand{\mll}[1]{\log^2 \Big({#1} \Big)}
\newcommand{\mlll}[1]{\log^3 \Big({#1} \Big)}

\title{Secretary Matching With Vertex Arrivals and No Rejections}
\author{
    Mohak Goyal
}
\affiliations{
  Department of Management Science \& Engineering, Stanford University\\mohakg@stanford.edu
}

\begin{document}

\maketitle

\begin{abstract}
Most prior work on online matching problems 
has been with the flexibility of keeping some vertices unmatched. We study three related online matching problems with the constraint of matching every vertex, i.e., \textit{with no rejections.} We adopt a model in which vertices arrive in a uniformly random order and the non-negative edge-weights are arbitrary. 
For the capacitated online bipartite matching problem, in which the vertices of one side of the graph are offline and those of the other side arrive online, we give a $4.62$-competitive algorithm when the capacity of each offline vertex is $2$. For the online general (non-bipartite) matching problem, where all vertices arrive online, we give a $3.34$-competitive algorithm. We also study the online roommate matching problem \cite{huzhang2017online}, in which each room (offline vertex) holds $2$ persons (online vertices). Persons derive non-negative additive utilities from their room as well as roommate. 
In this model, with the goal of maximizing the social welfare, we give a $7.96$-competitive algorithm. This is an improvement over the $24.72$ approximation factor in \cite{huzhang2017online}.
\end{abstract}
\section{Introduction}
 Online allocation problems study scenarios in which the input information is presented in steps, and the algorithm must make irreversible decisions at each step. These problems have applications in various areas such as for ride-hailing platforms \cite{dickerson2018allocation} and ad auctions \cite{mehta2012online}. In each of these applications, the future inputs are unknown and the goal is typically to maximize the revenue over the entire time horizon. Mathematically, several online allocation problems can be modeled as the problem of finding maximum-weight matchings on graphs.
 
 One of the most commonly studied models of online matching is with \emph{vertex arrivals}. In this model, one vertex arrives in each time step and it reveals weights of edges from itself to the previously arrived and offline vertices. The algorithm is required to decide whether the current vertex should be matched and if yes, to which other vertex. This model is adopted in several seminal papers in online matching theory, for example~\cite{karp1990optimal,kesselheim2013optimal,gamlath2019online}.  Another well-studied model is one in which all the vertices are known at the outset and edges are revealed in an online fashion. The decision to include an edge in the matching has to be made when it is seen, before the next edge is observed. This is referred to as the edge arrival model \cite{korula2009algorithms}. In this paper, we study problems in the vertex arrival model. Further, we assume that all edge-weights are non-negative.

The models of online matching problems also vary in the form of stochasticity in the edge weights and in the arrival order of vertices. It is easy to see that if both the edge weights and the arrival order of vertices are arbitrary then any online algorithm cannot be competitive against its offline counterpart. Therefore, the commonly studied models are of one of two forms:  \emph{secretary matching} \cite{dynkin1963optimum,gnedin1994solution} and prophet inequalities \cite{krengel1977semiamarts}. In the secretary matching model, vertices arrive in a uniformly random order but the edge-weights are arbitrary. Whereas, in the prophet inequalities setting, the edge weights are drawn independently from a known distribution but the vertex arrival order is arbitrary. In this paper, we consider the secretary matching model.

The secretary matching model of online matching is inspired by the classical secretary problem, in which a known number of job applicants arrive in a uniformly random order, are interviewed upon arrival, and at most one of them is  hired. The irreversible decision to hire or reject an applicant is taken before the next applicant arrives. The goal is to maximize the probability of hiring the best applicant. The optimal algorithm is $e$-competitive. The problem was a folklore before being solved formally by \cite{dynkin1963optimum}. See \cite{ferguson1989solved} for historical details. The lower bound of factor $e$ was given by \cite{gnedin1994solution}. 

Most known algorithms for problems in the secretary matching model use an explore-and-exploit approach. No matches are made in the exploration phase. In the exploitation phase, the current vertex is matched only if it has an edge with weight above a threshold \cite{korula2009algorithms}. Alternatively, as in \cite{ kesselheim2013optimal}, it is matched only if it is in a locally optimal matching computed over the vertices observed so far. An edge can be added to a matching only if it satisfies matching constraints, i.e., none of its end points have been matched previously by the algorithm.

All such algorithms make good use of the option of rejecting (i.e., not matching) several or all of the arriving vertices. This option may not be available in many resource allocation settings. Examples include the roommate market studied in \cite{chan2016assignment, huzhang2017online} and mentor-mentee matching in education programs. In ride-hailing too, rejecting customers is a costly option for platforms and may be chosen only in the worst situations. Therefore, in this paper we study three online matching problems with the constraint of not rejecting any vertex. 

The first problem is \emph{capacitated bipartite matching}
, in which vertices on one side of the graph are offline and have a fixed capacity and those on the other side arrive online and have capacity 1. We adopt an explore-and-exploit strategy which is given in detail in Section~\ref{sec:bm}.  
The second problem is \emph{general (non-bipartite) matching}  in which all vertices arrive online. A vertex can either be matched on arrival, or can be kept \emph{waiting} to be matched with a vertex that arrives later. It is not allowed to match two waiting vertices. For this problem, our algorithm runs in three phases:  \emph{exploration, selective matching}, and \emph{forced matching}. It is inspired by the two-phase algorithm of \cite{ezra2020secretary} for the general matching problem without the no-rejection condition. Details of our algorithm are given in Section~\ref{sec:gm}. 
The third problem is \emph{roommate matching} in Section~\ref{sec:rm}, in which there are $n/2$ offline rooms, each with capacity $2,$ and $n$ persons that arrive online must be assigned a room each. Persons derive utility from their room as well as roommate. 
 
\subsection{Related Work} 

Some extensions of the secretary problem include allowing multiple choices \cite{freeman1983secretary, preater1994multiple, gilbert2006recognizing, kleinberg2005multiple}, hiring a senior and a junior secretary \cite{preater1993senior}, and hiring a team of $k$ persons with submodular valuations over teams \cite{bateni2013submodular}. The secretary problem was generalized to matroids by \cite{babaioff2007matroids}. It is an open problem to find a constant factor approximation algorithm on general matroids. The current best is a $O(\log \log rank)$-competitive algorithm by \cite{lachish2014log}. 
Other works on the matroid secretary problem include \cite{ soto2013matroid,im2011secretary, gharan2013variants,
soto2021strong}.

 Secretary matching on bipartite graphs was introduced by \cite{korula2009algorithms}. They gave an $8$-competitive 
  algorithm. 
  The problem was resolved by \cite{kesselheim2013optimal} who gave an optimal $e$-competitive algorithm. More recently, \cite{reiffenhauser2019optimal} gave a truthful mechanism for secretary matching on bipartite graphs that attains the same competitive factor of $e$. Secretary matching on general graphs was solved recently by \cite{ezra2020secretary}. They gave an optimal $2.40$-competitive algorithm, which is, surprisingly, even better than the best-possible on bipartite graphs.

Another related line of work is on resource sharing. The offline roommate market was studied in \cite{chan2016assignment}. For a 2-persons-per-room model, they showed that maximizing social welfare is NP-hard and gave constant factor approximation algorithms for it. The online roommate market in the secretary setting was studied by \cite{huzhang2017online} and they too gave constant factor approximation algorithms for social welfare maximization. 
\cite{bei2018algorithms} gave algorithms for assignment in ride-sharing. 
 \cite{li2020fair} consider fairness considerations in resource sharing in general and for dorm assignment in particular.
\subsection{Our Contributions}
We have the following main results for secretary matching problems with no rejection:

\begin{enumerate}
\item For online capacitated bipartite matching, we give a $4.62$-competitive algorithm when offline vertices have capacity $2$ and a $5.46$-competitive analysis of the algorithm of \cite{huzhang2017online} when the offline vertices have capacity $1.$ They gave a $6.18$-competitive factor analysis. 
\item  We give a $3.34$-competitive algorithm for online general (non-bipartite) matching.
\item For online roommate matching, we give  a $7.96$-competitive algorithm. This result is an improvement over the $24.72$ factor given by \cite{huzhang2017online}.
\end{enumerate}
All our algorithms run in polynomial time. The competitive analysis results hold in expectation, which is taken over the randomness in the arrival order and in the algorithm.
\section{Model and Preliminaries}
An online algorithm is said to be $c$-competitive if its output has expected weight (or utility) at least $\frac{ \textsc{OPT}}{c},$ where $\textsc{OPT}$ is the weight (or utility) of the optimal offline solution. Denote the weights of matching $M$ and edge $e$ by $w(M)$ and $w(e)$ respectively. The three problems we study in this paper are described separately in the following subsections.

\subsection{No-Rejection Capacitated Bipartite Matching} 
This problem is defined over bipartite graph $G = (L, R, W),$ where $L$ is the set of offline vertices, $R$ is the set of online vertices, and $W$ is the set of edge weights. The offline vertices are known from the outset and the online vertices arrive in a uniformly random order. Each offline vertex has a fixed known capacity, which is the number of online vertices it can match with. Denote the capacity of vertex $u \in L$ by $c(u).$ We consider the case where $\sum_{u \in L} c(u) = n,$ where $n$ is the number of online vertices. We assume that $G$ is a complete bipartite graph and that all edge weights are non-negative.

The objective is to maximize the weight of the matching. 
In this paper we study two important special cases of the capacities of offline vertices. The first case is $c(u) = 1 $ for all $ u \in L$ and the second case is $c(u) = 2 $ for all $ u \in L.$ 
We refer to the former as \textsc{BipartiteMatching1} and to the latter as \textsc{BipartiteMatching2}.

\subsection{No-Rejection General Matching} \label{subsec:gm}
In this problem, all $n$ vertices $v \in V$ of graph $G = (V,W)$ arrive online in a uniformly random order and $n$ is even\footnote{This is required due to the no-rejection condition. However, our algorithm and its competitive ratio analysis work also if $n$ is odd and any one vertex is allowed to be kept unmatched.}. We consider a complete graph with non-negative edge-weights given by $W.$  
The algorithm can match the current vertex to a waiting vertex, or keep it waiting. However, it is required to match all vertices by the end of the process. 
The objective is to maximize the weight of the constructed matching. We refer to this problem as \textsc{GeneralMatching}. 

\subsection{No-Rejection Roommate Matching}  \label{subsec:rm}
This problem was defined in~\cite{huzhang2017online} who studied it as an online version of the offline roommate market model proposed by~\cite{chan2016assignment}. In this problem, there are $m$ rooms modelled as offline vertices and each room has $2$ beds. $n = 2m$ people arrive online in a uniformly random order and each person must be assigned a room upon arrival. Each room is assigned to $2$ persons.

Upon arrival, each person reveals her \emph{room valuation} for each of the $m$ rooms and \emph{mutual utility} of being roommates with each of the persons who arrived before her. The mutual utility captures the sum of happiness of both persons sharing a room. All room valuations and the mutual utility for every pair of persons are non-negative. Define a \emph{room allocation} as a set of $m$ tuples $(r^i, v^i_1, v^i_2)$ where $r^i$ denotes a room and $v^i_1, v^i_2$ denote the persons assigned to room $r^i$.  The utility of tuple $(r^i, v^i_1, v^i_2)$ is the sum of $2$ room valuations for room $r^i$ made by $v^i_1$ and $v^i_2$ and the mutual utility of persons $v^i_1$ and $v^i_2.$ The total utility of a room allocation is the sum of utilities of all its $m$ constituent tuples. The objective is to maximize the utility of the room allocation. We refer to this problem as \textsc{RoommateMatching}.

In the following sections, we give algorithms and technical results for the three problems described in this section.
\section{Online Capacitated Bipartite Matching} \label{sec:bm}
In this section we give online algorithms, \textsc{Alg1} and \textsc{Alg2}, for problems \textsc{BipartiteMatching1} and \textsc{BipartiteMatching2} respectively. Both the algorithms run in two-phases and employ a explore-exploit strategy.

For ease of notation, we number the online vertices from $1$ to $n$ in the order that they arrive. We use the integer variable $v$ both as the number of a step and the name of the current vertex. Define $k$ to be the stopping point of the exploration phase in both \textsc{Alg1} and \textsc{Alg2}. 
For \textsc{Alg1}, we set $k = \floor{0.21n}$ and for \textsc{Alg2} we set $k = \floor{0.25n}.$ We give the technical results in the following subsections. 

\subsection{\textsc{BipartiteMatching1}}
\textsc{Alg1} for \textsc{BipartiteMatching1} is the same as in \cite{huzhang2017online} [Section 3.1], but with a different stopping point for the exploration phase, in which arriving online vertices are matched to uniformly randomly chosen unmatched offline vertices. When vertex $v$ arrives in the exploitation phase, the algorithm computes an optimal matching $M^v$ over all the offline vertices and online vertices from $1$ to $v$. Denote the edge incident on $v$ in $M^v$ by $e^v$ and the neighbor of $v$ in $M^v$ by $l(e^v).$ If $l(e^v)$ is available, then $v$ is matched to it, otherwise $v$ is matched to a uniformly randomly chosen available offline vertex. 

\begin{algorithm}[tb]
\caption{\textsc{Alg1} for \textsc{BipartiteMatching1}}
\label{alg:algorithm1}
\begin{algorithmic}[1] 
\STATE $R' \leftarrow \emptyset$  \hfill \small{$\triangleright$ \textit{Set of online vertices seen so far}}
\STATE $M \leftarrow \emptyset$ \hfill \small{$\triangleright$ \textit{Matching}}
\FOR{every arriving vertex $v$}
\STATE $R' \leftarrow R' \cup \{v\}$
\IF [\hfill \small{$\triangleright$ \textit{Exploration phase}}] {$v < \floor{0.21n}$}
\STATE Uniformly randomly pick an available $v' \in L$
\STATE $M \leftarrow M \cup \{(v',v)\}$
\ELSE [\hfill \small{$\triangleright$ \textit{Exploitation phase}}]
\STATE $M^v \leftarrow$ Optimal bipartite matching on $G(L,R')$
\STATE $e^v \leftarrow$ Edge incident on $v$ in $M^v$
\IF {$M \cup \{e^v\}$ is a matching}
\STATE $M \leftarrow M \cup \{e^v\}$
\ELSE 
\STATE Uniformly randomly pick an available $v' \in L$
\STATE $M \leftarrow M \cup \{(v',v)\}$
\ENDIF
\ENDIF
\ENDFOR
\STATE \textbf{return} $M$
\end{algorithmic}
\end{algorithm}
\begin{lemma}  \label{lem:exp-wt-bm1}
For \textsc{BipartiteMatching1}, let \textsc{OPT} be the offline optimum value. The expected weight of edge $e^v$ computed in line $10$ of \textsc{Alg1} is at least $ \frac{\textsc{OPT}}{n}.$
\end{lemma}
\begin{proof}
This result follows directly from Lemma 1 of \cite{kesselheim2013optimal}. 
In any step $v$ of the algorithm, the identity and order of the  vertices arrived so far can be modelled via the following random process: first choose a set $R'$ of size $v$ from $R$. Then determine the arrival order of these $v$ vertices by iteratively selecting a vertex at random from $R'$ without replacement. 
In step $v,$ the algorithm calculates a optimal matching $M^v$ on $G[L, R']$. Since the current vertex $v$ can be seen as being selected uniformly at random from the set $R'$, the expected weight of the edge $e^v$ in $M^v$ is $w(M^v)/v$. Also, since $R'$ can be seen as being uniformly selected from $R$ with size $v$ we know $\mathbb{E}[w(M^v)] \geq \frac{v}{n} \textsc{OPT}$. Together we have, $\mathbb{E}[w(e^v)] \geq \frac{\textsc{OPT}}{n} .$
\end{proof}
\begin{lemma} \label{lem:prob-bm1}
For \textsc{BipartiteMatching1,} probability that edge $e^v$ computed in line $10$ of \textsc{Alg1} can be added to the matching $M$, i.e., the `if' condition of line $11$ in \textsc{Alg1} is 	`true' is at least $\frac{k}{v-1} \frac{n-v}{n} \left(1+ \frac{k}{n} \ml{ \frac{v}{k}}  + \frac{k^2}{2n^2} \mll{ \frac{v}{k}} -o(1) \right).$
\end{lemma}
\begin{proof}
Denote the probability that $e^v$ can be added to the matching by $\P(v_{\texttt{success}}).$ We start with a weak lower bound on $\P(v_{\texttt{success}}).$ For $l(e^v)$ to be available at step $v,$ it must neither be same as $l(e^u)$ for any $u < v$ nor must it been picked in a random matching step before. There are at most $v-1$ random matching steps before step $v,$ therefore the probability of  $l(e^v)$ being picked in a random matching is at most $\frac{v-1}{n}.$ We use the randomness of the arrival order to bound the probability of $l(e^v)$ being same as $l(e^u)$ for any $u < v.$ Consider a step $u <v$. Out of the $u$ participating online vertices in $M^u,$ the probability that $l(e^v)$ is matched to vertex $u$ is at most $\frac{1}{u}.$ This is because of the uniformly random order of the $u$ participating vertices. Further, this is independent of the order of the vertices $1, \ldots, u-1.$ Therefore, the event that for some $u' < u,$ $l(e^{u'}) = l(e^v)$  is independent of the event $l(e^{u}) = l(e^v)$. Following inductively from steps $v-1$ to $k+1,$ the probability that $l(e^v)$ was \emph{not} matched to $u$ in matching $M^u$ prior to step $v$ is:
\begin{align}
 \prod_{u = k+1}^{v-1} \P[l(e^u)  \neq  l(e^v)] \geq   \prod_{u = k+1}^{v-1} \bigg(1-\frac{1}{u} \bigg) = \frac{k}{v-1}.  \label{eq:not-in-matching-bm1}
\end{align}
Together with the bound on the probability of being picked in a random matching, this implies, 
\begin{align}
\P(v_{\texttt{success}}) \geq \bigg(1 - \frac{v-1}{n} \bigg)\bigg( \frac{k}{v-1}\bigg) \geq \frac{k(n-v)}{n(v-1)}. \label{eq:weak-bound-bm1}
\end{align}
Now we improve the bound in Equation~\eqref{eq:weak-bound-bm1} to get the result in the Lemma. Equation~\eqref{eq:weak-bound-bm1} implies that in any step $u >k,$ the probability that a random match is done is at most $1 - \frac{k(n-u)}{n(u-1)}.$ This is because the algorithm resorts to a random matching only if $l(e^u)$ is unavailable. If there is a random matching in step $u,$ then the probability that $l(e^v)$ is picked in it is $\frac{1}{n-(u-1)}$ because there are $n-(u-1)$ vertices available. Therefore, the probability that $l(e^v)$ is \emph{not} picked in a random matching is at least:
  \begin{align}
  &\prod_{u = 1}^k \left(1 - \frac{1}{n-(u-1)} \right) \prod_{u=k+1}^{v-1} \left(1 -\frac{1-\frac{k(n-u)}{n(u-1)}}{n-(u-1)} \right), \nonumber \\
  %
  &\geq \prod_{u = 1}^k  \frac{n-u}{n-(u-1)}   \prod_{u=k+1}^{v-1} \frac{ n-(u-1) -1+\frac{k(n-u)}{n(u-1)}}{n-(u-1)}, \nonumber \\
  &= \prod_{u = 1}^k \frac{n-u}{n-(u-1)} \prod_{u=k+1}^{v-1} \frac{ (n-u) \left[1 + \frac{k}{n(u-1)}\right]}{n-(u-1)}, \nonumber \\
   &= \prod_{u = 1}^{v-1} \frac{n-u}{n-(u-1)}  \prod_{u=k+1}^{v-1}  \left[1 + \frac{k}{n(u-1)}\right], \nonumber \\
&\geq \frac{n-v}{n}   \Bigg[1 + \sum_{u=k+1}^{v-1} \frac{ k}{n(u-1)} +  \frac{1}{2} \Big( \sum_{u=k+1}^{v-1} \frac{ k}{n(u-1)}\Big)^2 \nonumber \\
   & - \frac{1}{2}\bigg( \sum_{u=k+1}^{v-1} \frac{ k^2}{n^2(u-1)^2}\bigg)   \Bigg],  \label{eq:prod-to-sum}\\
   &\geq \frac{n-v}{n}   \left(1 + \frac{k }{n} \log \big(\frac{v}{k}\big) +  \frac{k^2}{2n^2}\log^2\big(\frac{v}{k}\big) -o(1) \right). \label{eq:not-randomly-matched-bm1}
  \end{align}
Equation~\eqref{eq:prod-to-sum} takes only the first $3$ terms in the expansion of the product. Equation~\eqref{eq:not-randomly-matched-bm1} follows by considering~\eqref{eq:prod-to-sum} as a Riemann sum with intervals of length $1$ and choosing appropriate lower bound in each interval of the sum. The $-o(1)$ term captures $\sum_{u=k+1}^{v-1} \frac{ -k^2}{2n^2(u-1)^2}$ and also the approximation of $\log(v-1)$ as $\log(v).$ This is possible because $v>k$ and $k$ is a constant fraction of $n.$    Combining Equations~\eqref{eq:not-randomly-matched-bm1} and~\eqref{eq:not-in-matching-bm1}, we get the desired result.
\end{proof}

\begin{theorem} \label{thm:bm1}
\textsc{Alg1} is $5.46$-competitive for \textsc{BipartiteMatching1}.  
\end{theorem}
\begin{proof}
We sum over the expected contributions of edges $e^v$ to the weight of the matching $M$  for all $v \in \{k+1, \ldots,n\}.$
\begin{align}
& \frac{\E(w(M))}{\textsc{OPT}} \geq  \sum_{v = k+1} ^n \frac{ \E(w(e^v))}{\textsc{OPT}}\P(v_{\texttt{success}}), \nonumber 
\end{align}
\begin{align}
&\geq \!\!\sum_{v = k+1}^n \!\frac{ 1}{n}\frac{k(n-v)}{n(v-1)} \Big(1+ \frac{k \log(\frac{v}{k})}{n}  + \frac{k^2\log^2(\frac{v}{k})}{2n^2}  -o(1) \Big), \nonumber\\ 
%
%
&\geq \frac{k }{n^2} \int_{k}^n (\frac{n-v}{v-1}) (1+ \frac{k \log(\frac{v}{k})}{n}  + \frac{k^2\log^2(\frac{v}{k})}{2n^2}  -o(1) ) dv, \nonumber \\
&=  \frac{k}{n^2} \Big[ n \ml{\frac{n}{k}}  + \frac{k}{2} \mll{\frac{n}{k}} - (n-k) \label{eq:alg1-opt} \\
& - \frac{k}{n}\Big( n \ml{\frac{n}{k}} -n +k \Big)  + \frac{k^2}{6n^2} \mlll{\frac{n}{k}} \nonumber \\
& -\frac{k^2}{2n^2} \Big( n \mll{\frac{n}{k}} -2n\ml{\frac{n}{k}} +2n-2k \Big) -o(n) \Big]. \nonumber 
\end{align}
The second inequality results from Lemmas~\ref{lem:exp-wt-bm1} and~\ref{lem:prob-bm1}. To get the third inequality, we view the integral $\int_{k}^n$   as a Riemann sum with subdivisions into intervals of length $1$ and use a simple upper bound on the function in each subdivision. The expression in Equation~\eqref{eq:alg1-opt} is maximized at $k = \floor{0.21n}$ and attains the value $0.1833$ for large $n$ which is $> \frac{1}{5.46}$ . 
\end{proof}

\subsection{\textsc{BipartiteMatching2}}
One natural algorithm for \textsc{BipartiteMatching2} is to consider each offline vertex (with capacity $2$) as two distinct vertices, each with capacity $1,$ and run \textsc{Alg1}. However, using the proposed \textsc{Alg2}, we utilize the fact that we have $2$ chances to match each offline vertex, and we do not need to do random matchings unless both chances are used. This is also reflected in the competitive ratio of \textsc{Alg2}, which is at most $4.62$ and is better than what we obtain for \textsc{Alg1}.

 \textsc{Alg2} ensures that no offline vertex is matched twice by random assignments until all offline vertices are matched at least once. \textsc{Alg2} starts with an exploration phase in which only random matches are made. In the exploitation phase, it finds a match for arriving vertex $v$ via an optimal matching $M^v,$ and by a random selection if that match is not feasible.

The technical lemmas on the probability of availability of $l(e^v)$ follow and the competitive ratio is given in Theorem~\ref{thm:bm2}.

\begin{algorithm}[tb]
\caption{\textsc{Alg2} for \textsc{BipartiteMatching2}}
\label{alg:algorithm2}
\begin{algorithmic}[1] 
\STATE $R' \leftarrow \emptyset$  \hfill \small{$\triangleright$ \textit{Set of online vertices seen so far}}
\STATE $M \leftarrow \emptyset$ \hfill \small{$\triangleright$ \textit{Matching}}
\STATE $L_a \leftarrow L$ \hfill \small{$\triangleright$ \textit{Offline vertices available for random matchings}}
\FOR{every arriving vertex $v$}
\STATE $R' \leftarrow R'\cup \{v\}$
\IF [\hfill \small{$\triangleright$ \textit{Exploration phase}}]{$ v< \floor{n/4}$}
\STATE Uniformly randomly pick a vertex $v' \in L_a$
\STATE $M \leftarrow M \cup \{(v',v)\}$
\STATE $L_a \leftarrow L_a \setminus \{v'\}$
\ELSE [\hfill \small{$\triangleright$ \textit{Exploitation phase}}]
\STATE $M^v \leftarrow$ Optimal capacitated matching on $G(L,R')$
\STATE $e^v \leftarrow$ Edge incident on $v$ in $M^v$
\STATE $l(e^v) \leftarrow$ Neighbor of $v$ via $e^v$
\IF {$M \cup \{e^v\}$ is a valid capacitated matching}
\STATE $M \leftarrow M \cup \{e^v\}$
\STATE $L_a \leftarrow L_a \setminus \{l(e^v)\}$
\ELSE
\STATE Uniformly randomly pick a vertex $v' \in L_a$
\STATE $M \leftarrow M \cup \{(v',v)\}$
\STATE $L_a \leftarrow L_a \setminus \{v'\}$
\ENDIF
\ENDIF
\IF [\hfill \small{$\triangleright$ \textit{Reset $L_a$}}] {$L_a= \emptyset$}
\STATE $L_a \leftarrow \{v' \in L | v'~\text{is not matched to full capacity}\}$
\ENDIF
\ENDFOR
\STATE \textbf{return} $M$
\end{algorithmic}
\end{algorithm}
\begin{lemma} \label{lem:exp-wt-bm2}
For \textsc{BipartiteMatching2,} let \textsc{OPT} be the offline optimum value. The expected weight of edge $e^v$ computed in line $12$ of \textsc{Alg2} is at least $ \frac{\textsc{OPT}}{n}.$
\end{lemma}
\begin{proof}
Notice that the edge $e^v$ is computed same as in \textsc{Alg1}. Make $c$ identical copies of each offline vertex, each with capacity $1.$ The problem reduces to \textsc{BipartiteMatching1} and this result follows from Lemma~\ref{lem:exp-wt-bm1}.
\end{proof}

\begin{lemma} \label{lem:prob1-bm2}
In \textsc{Alg2} for \textsc{BipartiteMatching2,} for vertices index $k+1$ to $\floor{n/2},$ the probability that edge $e^v$ computed in line $12$ can be added to matching $M$, i.e., the `if' condition of line $14$ is `true' is at least $\frac{n-v}{2v} - \frac{n^2}{32v^2} -o(1).$
\end{lemma}

\begin{lemma} \label{lem:prob2-bm2}
In \textsc{Alg2} for \textsc{BipartiteMatching2,} for vertices index $\floor{n/2}+1$ to $n,$ the probability that edge $e^v$ computed in line $12$ can be added to matching $M$, i.e., the `if' condition of line $14$ is `true' is at least $\frac{3n(n-v)}{16v^2} -o(1).$
\end{lemma}
The proofs of Lemmas~\ref{lem:prob1-bm2} and~\ref{lem:prob2-bm2} use similar ideas as the proof of Lemma~\ref{lem:prob-bm1} and  are given in Appendix~A.
\begin{theorem}  \label{thm:bm2}
\textsc{Alg2} is $4.62$-competitive for \textsc{BipartiteMatching2}.  
\end{theorem}
\begin{proof}
We sum the expected contributions of the edges $e^v$ to the weight of the matching $M$ for all $ v \in \{ k+1, \ldots , n\}. $
From Lemmas~\ref{lem:exp-wt-bm2},~\ref{lem:prob1-bm2}, and~\ref{lem:prob2-bm2}, we get:
\begin{align}
& \frac{\E(w(M))}{\textsc{OPT}} \geq  \frac{1}{\textsc{OPT}}\sum_{v = \floor{\frac{n}{4}}+1} ^n \E(w(e^v))\P(v_{\texttt{success}}) \nonumber \\
&\geq \frac{ 1}{n} \!\! \sum_{v = \floor{\frac{n}{4}}+1}^{n/2}\!\!\!\! \bigg(\frac{n-v}{2v} - \frac{n^2}{32v^2} \bigg)  + \frac{1}{n} \!\sum_{v = n/2}^{n}\!\! \frac{3n(n-v)}{16v^2} -o(1) \nonumber \\
&\geq  \int_{\frac{n}{4}}^{\frac{n}{2}} \bigg(\frac{n-v}{2vn} - \frac{n^2}{32v^2n} \bigg)dv + \int_{\frac{n}{2}}^{n} \frac{3(n-v)}{16v^2}  dv  -o(1) \nonumber \\
\intertext{We do a change of variables by substituting $v$ with $xn$,}
&= \int_{1/4}^{1/2} \frac{1-x}{2x}  -\frac{1}{32x^2} dx  + \int_{1/2}^{1} \frac{3(1-x)}{16x^2}  dx - o(1), \nonumber\\
&> 0.2166> \frac{1}{4.62} .\nonumber
\end{align}
The $o(1)$ error term captures the difference between $k = \floor{n/4}$ and $n/4$ in the limit of the integration; this difference goes to $0$ as $n \rightarrow \infty.$
\end{proof}
%
%
%
%
\section{ Online General Matching}  \label{sec:gm}
For this problem, described in Subsection~\ref{subsec:gm}, we give \textsc{Alg3} which runs in three phases and is inspired by \cite{ezra2020secretary} who have a two-phase algorithm without the no-rejection condition. We number the vertices from $1$ to $n$ in the order of arrival. We use variable $v$ both as the number of an iteration and as the name of the current vertex.

 The first phase is of exploration in which all vertices are kept waiting and added to set $A$. Next is the selective matching phase. When vertex $v$ arrives in this phase, if $v$ is even, the algorithm computes an optimal matching $M^v$ over the set of vertices $V'$ seen so far. If $v$ is odd, the algorithm uniformly randomly chooses a vertex $v_r$ from $V' \setminus {v}$ , and computes the optimal matching $M^v$ over $V' \setminus {v_r}$. This ensures that $v$ has a match in $M^v$. If the match of $v$ in $M^v$ (denoted by $l(e^v)$) is a waiting vertex (i.e., it is in $A$), then they are matched in $M$, otherwise $v$ is kept waiting and added to $A$. The third phase is called forced matching. In this phase, the algorithm computes an optimal match $M^v$ just as in the second phase. If $l(e^v)$ is in $A,$ then $v$ is matched to it, otherwise $v$ is matched to a randomly chosen waiting vertex $v' \in A$. 

 Define $k_e$ and $k_s$ as the stopping points of the first and second phases respectively in \textsc{Alg3}. We set $k_e = \floor{\frac{6n}{17}}$ and $k_s$ is not fixed. The second phase ends when the number of waiting vertices, i.e., $|A|$ equals the number of vertices yet to arrive, i.e., $n+1-v$. In Lemma~\ref{lem:conc-eq} we show  that $k_s$ is not much smaller than its expected value $ \frac{12n}{17}$ w.h.p. for large $n.$
\begin{algorithm}[tb]
\caption{\textsc{Alg3} for \textsc{GeneralMatching}}
\label{alg:generalmatching}
\begin{algorithmic}[1] 
\STATE $V' \leftarrow \{1,2, \ldots \floor{6n/17}\}$ \hfill \small{$\triangleright$ \textit{Set of observed vertices}}
\STATE $A \leftarrow \{1,2, \ldots \floor{6n/17}\}$ \hfill \small{$\triangleright$ \textit{Set of waiting vertices}}
\STATE $M \leftarrow \emptyset$ \hfill \small{$\triangleright$ \textit{Matching}}
\FOR{every arriving vertex $v > \floor{6n/17}$}
\STATE $V' \leftarrow V'\cup \{v\}$
\IF {$v$ is even}
\STATE $\tilde{V} = V'$
\ELSE
\STATE $v_r \leftarrow$ randomly chosen from $V'\setminus \{v\}$
\STATE $\tilde{V} = V' \setminus \{v_r\}$
\ENDIF
\STATE $M^v \leftarrow$ Optimal matching on $\tilde{V}$
\STATE $e^v \leftarrow$ Edge incident on $v$ in $M^v$
\STATE $l(e^v) \leftarrow$ Neighbor of $v$ via $e^v$
\IF {$M \cup \{e^v\}$ is a valid matching}
\STATE $M \leftarrow M \cup \{e^v\}$
\STATE $A \leftarrow A \setminus \{l(e^v)\}$ \hfill \small{$\triangleright$ \textit{$l(e^v)$ is no longer waiting}}
\ELSE
\IF [\hfill \small{$\triangleright$ \textit{Keep $v$ waiting}}]{$|A| < (n+1-v)$ } 
\STATE $A \leftarrow A \cup \{v\}$
\ELSE [\hfill \small{$\triangleright$ \textit{Cannot keep any more vertices waiting}}]
\STATE Uniformly randomly pick a vertex $v' \in A$
\STATE $M \leftarrow M \cup \{(v',v)\}$ \hfill \small{$\triangleright$ \textit{Forced matching}}
\STATE $A \leftarrow A \setminus \{v'\}$ 
\ENDIF
\ENDIF
\ENDFOR
\STATE \textbf{return} $M$
\end{algorithmic}
\end{algorithm}
\begin{lemma} \label{lem:exp-wt-gm}
 For \textsc{GeneralMatching,} let \textsc{OPT} be the offline optimum value. The expected weight of edge $e^v$ computed in line $13$ of \textsc{Alg3} is at least $ \frac{4\floor{v/2}-2}{n(n-1)}\textsc{OPT}.$
\end{lemma}
\begin{proof}
This result  follows  from Theorem 3.1 of \cite{ezra2020secretary}. See that $V'$ is a uniformly randomly sampled subset of $V$ of size $v$ and therefore the expected weight of an edge with both end-points in $V'$ is equal to the average weight of edges with end-points in $V.$ Further, $e^v$ can be seen as uniformly randomly sampled from the edges in $M^v$.
\end{proof}

\begin{lemma} \label{lem:prob1-gm}
In \textsc{Alg3} for \textsc{GeneralMatching,} for vertices index $k_e+1$ to $k_s,$ the probability that edge $e^v$ computed in line $13$ can be added to the matching, i.e., the `if' condition of line $15$ is `true,' is at least $\frac{1}{3}(1+\frac{2(k_e-2)^3}{(v-1)^3}).$
\end{lemma}
\begin{proof}
This result follows directly from Lemma 3.2 of \cite{ezra2020secretary}. They do an elaborate accounting of the probability that a vertex is matched by the time that vertex $v$ arrives, conditioned on the set of vertices arrived so far. 
\end{proof}

\begin{lemma} \label{lem:prob2-gm}
In \textsc{Alg3} for \textsc{GeneralMatching,} for vertices  $k_s+1$ to $n,$ the probability that edge $e^v$ computed in line $13$ can be added to the matching, i.e., the `if' condition of line $15$ is `true,' is at least $\frac{1}{3}(1+\frac{2(k_e-2)^3}{(v-1)^3})(\frac{n-v+1}{n-k_s}).$
\end{lemma}
\begin{proof}
In the forced matching phase, consider two processes running independently. The first process is that of finding a vertex $l(e^v)$ for $v$, as computed in lines $11$ and $12$ of \textsc{Alg3} and starts at step $k_e+1$. The second process is selecting a vertex $v'$ from the set of waiting vertices $A$ uniformly at random and starts at step $k_s+1$. Vertex $v'$ is matched to $v$ if $l(e^v)$ is not available. 
We are interested in the probability of $l(e^v)$ being available when $v$ arrives. For a lower bound of this probability, it is sufficient to consider the case that $l(e^v)$ was not picked earlier by either of the two processes. The probability of it not being picked in the first process (i.e., in optimal matchings $M^u$ for $u <v$) is at least $\frac{1}{3}(1+\frac{2(k_e-2)^3}{(v-1)^3})$ by Lemma~\ref{lem:prob1-gm}. The probability of $l(e^v)$ not being picked in the second process (random selection) is at least $(1 - \frac{v-1-k_s}{n-k_s}),$ since there are at most $(v-1-k_s)$ vertices picked in this process out of $n-k_s$ vertices. This expression simplifies to $\frac{n-v+1}{n-k_s}.$ The result follows from multiplying these two probabilities.
\end{proof}

\begin{lemma} \label{lem:conc-eq}
For \textsc{Alg3}, the stopping point of the selective matching phase, $k_s,$ is at least $\left(\frac{12-\delta}{17}\right)n$ for any small positive constant $\delta$ w.h.p. as $n \rightarrow \infty.$ 
\end{lemma}
\begin{proof}
Denote the set of steps in the selective matching phase, i.e., $ \{k_e+1, \ldots, k_s\}$ by $V_{\texttt{SM}}.$ Further, denote the set of steps $ \{k_e+1, \ldots, t\}$ where $t \leq k_s$ by $V^t_{\texttt{SM}}.$ Denote the indicator random variable of the event that \textsc{Alg3} keeps vertex $v \in V_{\texttt{SM}}$ waiting when it arrives by $X_v.$ 
 By Lemma~\ref{lem:prob1-gm},
\begin{align}
\P(X_v =1) \leq 1 - \frac{1}{3}(1+\frac{2(k_e-2)^3}{(v-1)^3}) = \frac{2}{3} ( 1 - \frac{(k_e-2)^3}{(v-1)^3} ). \nonumber
\end{align}
Clearly, $X_v$ for $v \in V^t_{\texttt{SM}}$ are not independent random variables and therefore we cannot use a Chernoff bound on their sum. However, to obtain a useful concentration inequality, we prove the following property: for any subset $S \subseteq V_{\texttt{SM}},$
\begin{align}
\P \bigg[\big(\prod_{u \in S} X_u\big) = 1 \bigg] \leq \prod_{u \in S} \P(X_u =1). \label{eq:neq-rel-property}
\end{align}
To prove Equation~\eqref{eq:neq-rel-property}, we analyze how the random variables $X_u$ and $X_{u'}$ depend on each other for any pair of steps $(u,u')$. Without loss of generality, let $u' < u.$ See that $X_{u'}$ is independent of $X_{u}$ because it is observed before step $u.$ For dependence of $X_{u}$ on $X_{u'}$, consider the case that $X_{u'} =1.$ This event adds vertex $u'$ to the set of waiting vertices $A$. Since this event does not remove any vertex from $A$, it does not increase the probability that $l(e^u)$ is not in $A$. That is, $\P( X_u = 1|X_{u'} =1) \leq \P(X_u=1).$ Extending this argument to all vertices in $S$ that arrived before $u$, conditioned on the event $X_{u'} = 1$ for all $\{u' \in S| u' < u\}$, the probability that $u$ is kept waiting is no larger than the unconditional probability $\P(X_u =1).$ This proves that the property in Equation~\eqref{eq:neq-rel-property} holds. Intuitively, the process of vertices being kept waiting upon arrival is \emph{self-correcting} such that if too many vertices are kept waiting, then the probability of future vertices being kept waiting cannot increase.

Now we define $t$ independent Bernoulli random variables $y_u$ for $u \in V^t_{\texttt{SM}}$ such that $\P(y_u = 1)  = \P(X_u =1).$ Denote $X =\sum_{u \in V^t_{\texttt{SM}}} X_{u}$ and $Y = \sum_{u \in V^t_{\texttt{SM}}} y_u.$ Notice that $X$ and $Y$ are functions of $t;$ it is omitted from the notation for clarity. Clearly $\E(Y) = E(X).$ We will prove that for
any $a > 0,$
\begin{align}
\E[e^{aX}] \leq [e^{aY}]. \label{eq:exy}
\end{align}

Since  $e^{az}$ can be expanded as $\sum_{s = 0}^{+\infty} (a^s /s!) z^s,$ for any $z \in \mathbb{R},$ by the linearity of expectation, we have  $\E[e^{aX}] = \sum_{s = 0}^{+\infty} (a^s /s!) \E[X^s]$ and   $\E[e^{aY}] = \sum_{s = 0}^{+\infty} (a^s /s!) \E[Y^s].$ To prove eq.~\eqref{eq:exy}, it suffices to show that for every $s = 0,1,\ldots,\infty,$  $\E[X^s] \leq \E[Y^s].$ By the definition of $X,$ we have $X^s = \sum_\sigma \prod_{j=1}^s X_{\sigma(j)}$ where the summation is over all permutations $\sigma$ selecting $s$ items from $V^t_{\texttt{SM}}$ with replacement. By linearity of expectation, $\E[X^s] = \sum_\sigma \E[\prod_{j=1}^s X_{\sigma(j)}]$ and $\E[Y^s] = \sum_\sigma \E[\prod_{j=1}^s Y_{\sigma(j)}].$ To prove eq.~\eqref{eq:exy}, it now suffices to prove that for every permutation $\sigma,  \sum_\sigma \E[\prod_{j=1}^s X_{\sigma(j)}] \leq \E[\prod_{j=1}^s Y_{\sigma(j)}].$ Define $Q$ to be the image of $\sigma,$ that is $Q$ is the set of distinct elements $q$ such that $\sigma(j) = q$ for some $j.$ We have:
\begin{align}
&  \E\big[\prod_{j=1}^s X_{\sigma(j)}\big] = \P[X_u = 1 ~\forall ~u \in Q] \leq \prod_{u \in Q} \P(X_u =1)  \nonumber \\
&= \prod_{u \in Q} \P(Y_u =1) = \P[Y_u = 1~ \forall ~u \in Q] = \E\big[\prod_{j=1}^s Y_{\sigma(j)}\big].\nonumber  
\end{align}
The inequality follows from~\eqref{eq:neq-rel-property}. This proves~\eqref{eq:exy}. We apply Markov's inequality \cite{motwani1995randomized} to the non-negative random variable $e^{aX}$ for $a>0$ and $\varepsilon > 0$:
\begin{align}
&\P[X \geq (1+\varepsilon)\E[X]] =  \P[e^{aX} \geq e^{a(1+\varepsilon)\E[X]}] \nonumber\\
&\leq \E[e^{aX}]/e^{a(1+\varepsilon)\E[X]} \leq \E[e^{aY}]/e^{a(1+\varepsilon)\E[Y]}.  \label{eq:markov}
\end{align}
The last inequality is due to~\eqref{eq:exy} and the fact that $\E[X] = \E[Y].$ Now we use the standard steps for proving the Chernoff bound for $Y$. We use $a = \log (1+\varepsilon).$ Then,  
$$\frac{\E[e^{aY}]}{e^{a(1+\varepsilon)\E[Y]}} \leq \frac{e^{\E[Y](e^a-1)}}{e^{(1+\varepsilon)\E[Y]}} = \Big(\frac{e^\varepsilon}{(1+\varepsilon)^{1+\varepsilon}} \Big)^{\E[Y]} \!\leq e^{\frac{-E[Y]\varepsilon^2}{3}}.$$
Together with~\eqref{eq:markov}, and using $\E[X] = \E[Y]$, this implies:
\begin{align}
\P[X \geq (1+\varepsilon)\E[X]] \leq e^{\frac{-E[X]\varepsilon^2}{3}}. \label{eq:conc}
\end{align} 
We now use this concentration bound on $X$ to prove the result of the Lemma. 
Denote $z$ as the number of steps in $V_{\texttt{SM}}$ in which the current vertex is kept waiting. By definition of $k_s,$ the number of vertices waiting after step $k_s$ is equal to the number of vertices yet to arrive.  That is: $k_e + z - (k_s-k_e-z)  = n-k_s\implies z = \frac{n - 2k_e}{2}= \frac{5n}{34}.$ This implies that: for $k_s$ to be $< \frac{12-\delta}{17}n,$ we must have  $X = 5n/34$ for $t < \frac{12-\delta}{17}n.$ To be able to use Equation~\eqref{eq:conc}, we compute $\E[X]$ for $t = \frac{12-\delta}{17}n.$ This is, 
$$ \E[X] = \sum_{u \in V^t_{\texttt{SM}}} \P(v_{\texttt{wait}}) = \sum_{u = \floor{\frac{6n}{17}}+1}^{\frac{12-\delta}{17}n} \bigg[\frac{2}{3} \bigg( 1 - \frac{(k_e-2)^3}{(u-1)^3} \bigg)\bigg].$$
Upon solving, we get $\E[X] \leq (\frac{5}{34} - \frac{7\delta}{204})n + o(n).$
Therefore, for $\varepsilon = \delta/10$, we have $(1+\varepsilon)\E[X] < 5n/34.$ This, together with Equation~\eqref{eq:conc}, implies that for any constant $\delta >0,$ we get $k_s > \frac{12-\delta}{17}n$ w.h.p. as $n \rightarrow \infty.$
\end{proof}

\begin{theorem}  \label{thm:gm}
\textsc{Alg3} is $3.34$-competitive for \textsc{GeneralMatching} w.h.p. as $n \rightarrow \infty$.  
\end{theorem}
\begin{proof}
We sum the expected contributions of the edges $e^v$ to the weight of the matching $M$ for all $ v \in \{ k_e+1, \ldots , n\}. $
From Lemmas~\ref{lem:exp-wt-gm},~\ref{lem:prob1-gm}, and~\ref{lem:prob2-gm}, we get:
\begin{align}
&  \frac{\E(w(M))}{\textsc{OPT}} \geq \sum_{v = k_e+1}^n  \frac{\E(w(e^v))}{\textsc{OPT}} \P(e^v \text{is added to matching}),\nonumber \\
&\geq \frac{1}{n(n-1)} \Bigg[ \sum_{v = k_e+1}^{k_s} \frac{4\floor{v/2}-2}{3}\Big(1+\frac{2(k_e-2)^3}{(v-1)^3}\Big) \nonumber \\ 
&+   \sum_{v = k_s+1}^{n}  \frac{4\floor{v/2}-2}{3}\Big(1+\frac{2(k_e-2)^3}{(v-1)^3}\Big) \Big(\frac{n-v+1}{n-k_s}\Big) \Bigg],\nonumber \\
&\geq \frac{1}{n^2} \Bigg[ \sum_{v = k_e+1}^{k_s} \frac{2v-4}{3}\Big(1+\frac{2(k_e-2)^3}{(v-1)^3}\Big) \nonumber \\ 
&+   \sum_{v = k_s+1}^{n}  \frac{2v-4}{3}\Big(1+\frac{2(k_e-2)^3}{(v-1)^3}\Big) \Big(\frac{n-v+1}{n-k_s}\Big) \Bigg],\nonumber \\
&\geq \frac{1}{n^2} \Bigg[\int_{k_e}^{k_s} \frac{2v-4}{3}\Big(1+\frac{2(k_e-2)^3}{(v-1)^3}\Big) dv  \nonumber \\
&+ \int_{k_s}^{n}  \frac{2v-4}{3}\Big(1+\frac{2(k_e-2)^3}{(v-1)^3}\Big) \Big(\frac{n-v+1}{n-k_s}\Big) dv \Bigg]. \nonumber 
\end{align}
Now we substitute the values of the parameters $k_e$ and $k_s.$ The value of $k_e$ is $\floor{\frac{6n}{17}}.$ For any $v,$ the value of the second integrand is smaller than that of the first integrand. Therefore, a lower bound of the expression is obtained by setting $k_s$ to its smallest value. We had observed in Lemma~\ref{lem:conc-eq} that $k_s$ is at least $\left(\frac{12-\delta}{17}\right)n$ for any small positive constant $\delta$ with high probability for $n \rightarrow \infty.$ We omit further calculations due to space constraints. On solving the integration, we get:
 $$  \frac{\E(w(M))}{\textsc{OPT}} \geq 0.30005 - \frac{10 \delta}{289} - o(1).  $$
By Lemma~\ref{lem:conc-eq}, the bound on $k_s$ holds for any positive constant $\delta.$ Using $\delta = 10^{-5},$ we get the factor $> \frac{1}{3.34}$.
\end{proof}

\section{Online Roommate Matching} \label{sec:rm}
Recall the online roommate matching problem given in Subsection~\ref{subsec:rm}. 
If the mutual utilities of all pairs of persons are $0,$ then this problem is the same as \textsc{BipartiteMatching2} with rooms as offline vertices, persons as online vertices, and room valuations as edge-weights. Whereas, if all room valuations are $0$, then the problem is the same as \textsc{GeneralMatching} over persons as online vertices and mutual utilities as edge-weights. With these observations, we give \textsc{Alg4} for  \textsc{RoommateMatching} which considers only the room valuations with probability $p$ and only the mutual utilities with probability $1-p.$ The competitive ratio in our analysis is minimized at $p = 0.58.$
\begin{algorithm}[tb]
\caption{\textsc{Alg4} for \textsc{RoommateMatching}}
\label{alg:roommate}
\begin{algorithmic}[1] 
\STATE Draw a random variable $r$ from Uniform$[0,1]$
\IF {$r \leq 0.58$}
\STATE Run \textsc{Alg2} on room valuations
\ELSE
\STATE Run \textsc{Alg3} on mutual utilities
\ENDIF
\end{algorithmic}
\end{algorithm}
\begin{theorem}  \label{thm:rm}
\textsc{Alg4} is $7.96$-competitive for \textsc{RoommateMatching} w.h.p. as $n \rightarrow \infty$.  
\end{theorem}
\begin{proof}
Let $\textsc{OPT}$ be the social welfare of the optimal offline room allocation. Let $\textsc{OPT}_{RV}$ and $\textsc{OPT}_{MU}$ be the social welfare of the offline room allocations which maximize only the sum of room valuations and mutual utilities respectively.
 Then we have $ \textsc{OPT}_{RV} + \textsc{OPT}_{MU} \geq \textsc{OPT}.$  Let $U$ denote the expected social welfare of the room allocation given by \textsc{Alg4}. By Theorems~\ref{thm:bm2} and~\ref{thm:gm}, we have $U \geq 0.58 \cdot \frac{1}{4.62}  \cdot\textsc{OPT}_{RV} + 0.42 \cdot \frac{1}{3.34}  \cdot\textsc{OPT}_{MU} \geq  0.1257 \cdot ( \textsc{OPT}_{RV} + \textsc{OPT}_{MU}) \geq 0.1257 \cdot \textsc{OPT} > \frac{\textsc{OPT}}{7.96}.$ 
\end{proof}

\section{Conclusion}
In this paper we do the first detailed analysis of online matching problems with a no-rejection condition. We argue that this is a natural constraint in several resource allocation and resource sharing scenarios. We give constant factor approximation algorithms for capacitated bipartite matching, general matching, and roommate matching problems in the online no-rejection setting. The roommate matching problem captures scenarios where multiple persons may be assigned to use a single resource and there are positive externalities from the other persons using the same resource. 

For future work, an important theoretical direction is to find lower bounds of competitive ratios for these problems. Simple lower bounds follow from the corresponding problems without the no-rejection condition. This is $e (\approx 2.73)$ for \textsc{BipartiteMatching1} and $2.40$ for \textsc{GeneralMatching}. 
Another interesting problem is to design truthful mechanisms for these online matching and resource allocation problems. In the proposed algorithms, an arriving vertex has the incentive to misreport its valuations (i.e., edge-weights) to potentially get a better match. For example, it is a dominant strategy for an arriving vertex to report zero valuations for the resources that are no longer available. 

\newpage
\section{Acknowledgement}
I am very grateful to Ashish Goel, Zhihao Jiang, and Anmol Kagrecha for their thoughtful comments on this work.
\bibliography{references}

\begin{thebibliography}{29}
\providecommand{\natexlab}[1]{#1}

\bibitem[{Babaioff, Immorlica, and Kleinberg(2007)}]{babaioff2007matroids}
Babaioff, M.; Immorlica, N.; and Kleinberg, R. 2007.
\newblock Matroids, secretary problems, and online mechanisms.
\newblock In \emph{Proceedings of the eighteenth annual ACM-SIAM symposium on
  Discrete algorithms}, 434--443.

\bibitem[{Bateni, Hajiaghayi, and Zadimoghaddam(2013)}]{bateni2013submodular}
Bateni, M.; Hajiaghayi, M.; and Zadimoghaddam, M. 2013.
\newblock Submodular secretary problem and extensions.
\newblock \emph{ACM Transactions on Algorithms (TALG)}, 9(4): 1--23.

\bibitem[{Bei and Zhang(2018)}]{bei2018algorithms}
Bei, X.; and Zhang, S. 2018.
\newblock Algorithms for trip-vehicle assignment in ride-sharing.
\newblock In \emph{AAAI Conference on Artificial Intelligence}.

\bibitem[{Chan et~al.(2016)Chan, Huang, Liu, Zhang, and
  Zhang}]{chan2016assignment}
Chan, P.; Huang, X.; Liu, Z.; Zhang, C.; and Zhang, S. 2016.
\newblock Assignment and pricing in roommate market.
\newblock In \emph{AAAI Conference on Artificial Intelligence}, volume~30.

\bibitem[{Dickerson et~al.(2018)Dickerson, Sankararaman, Srinivasan, and
  Xu}]{dickerson2018allocation}
Dickerson, J.; Sankararaman, K.; Srinivasan, A.; and Xu, P. 2018.
\newblock Allocation problems in ride-sharing platforms: Online matching with
  offline reusable resources.
\newblock In \emph{AAAI Conference on Artificial Intelligence}, volume~32.

\bibitem[{Dynkin(1963)}]{dynkin1963optimum}
Dynkin, E.~B. 1963.
\newblock The optimum choice of the instant for stopping a Markov process.
\newblock \emph{Soviet Mathematics}, 4: 627--629.

\bibitem[{Ezra et~al.(2020)Ezra, Feldman, Gravin, and Tang}]{ezra2020secretary}
Ezra, T.; Feldman, M.; Gravin, N.; and Tang, Z.~G. 2020.
\newblock Secretary Matching with General Arrivals.
\newblock \emph{arXiv preprint arXiv:2011.01559}.

\bibitem[{Ferguson(1989)}]{ferguson1989solved}
Ferguson, T.~S. 1989.
\newblock Who solved the secretary problem?
\newblock \emph{Statistical science}, 4(3): 282--289.

\bibitem[{Freeman(1983)}]{freeman1983secretary}
Freeman, P. 1983.
\newblock The secretary problem and its extensions: A review.
\newblock \emph{International Statistical Review/Revue Internationale de
  Statistique}, 189--206.

\bibitem[{Gamlath et~al.(2019)Gamlath, Kapralov, Maggiori, Svensson, and
  Wajc}]{gamlath2019online}
Gamlath, B.; Kapralov, M.; Maggiori, A.; Svensson, O.; and Wajc, D. 2019.
\newblock Online matching with general arrivals.
\newblock In \emph{IEEE 60th Annual Symposium on Foundations of Computer
  Science (FOCS)}, 26--37.

\bibitem[{Gharan and Vondr{\'a}k(2013)}]{gharan2013variants}
Gharan, S.~O.; and Vondr{\'a}k, J. 2013.
\newblock On variants of the matroid secretary problem.
\newblock \emph{Algorithmica}, 67(4): 472--497.

\bibitem[{Gilbert and Mosteller(2006)}]{gilbert2006recognizing}
Gilbert, J.~P.; and Mosteller, F. 2006.
\newblock Recognizing the maximum of a sequence.
\newblock In \emph{Selected Papers of Frederick Mosteller}, 355--398. Springer.

\bibitem[{Gnedin(1994)}]{gnedin1994solution}
Gnedin, A.~V. 1994.
\newblock A solution to the game of googol.
\newblock \emph{The Annals of Probability}, 1588--1595.

\bibitem[{Huzhang et~al.(2017)Huzhang, Huang, Zhang, and
  Bei}]{huzhang2017online}
Huzhang, G.; Huang, X.; Zhang, S.; and Bei, X. 2017.
\newblock Online Roommate Allocation Problem.
\newblock In \emph{IJCAI}, 235--241.

\bibitem[{Im and Wang(2011)}]{im2011secretary}
Im, S.; and Wang, Y. 2011.
\newblock Secretary problems: Laminar matroid and interval scheduling.
\newblock In \emph{Proceedings of the twenty-second annual ACM-SIAM Symposium
  on Discrete Algorithms (SODA)}, 1265--1274.

\bibitem[{Karp, Vazirani, and Vazirani(1990)}]{karp1990optimal}
Karp, R.~M.; Vazirani, U.~V.; and Vazirani, V.~V. 1990.
\newblock An optimal algorithm for on-line bipartite matching.
\newblock In \emph{Proceedings of the twenty-second annual ACM Symposium on
  Theory of Computing (STOC)}, 352--358.

\bibitem[{Kesselheim et~al.(2013)Kesselheim, Radke, T{\"o}nnis, and
  V{\"o}cking}]{kesselheim2013optimal}
Kesselheim, T.; Radke, K.; T{\"o}nnis, A.; and V{\"o}cking, B. 2013.
\newblock An optimal online algorithm for weighted bipartite matching and
  extensions to combinatorial auctions.
\newblock In \emph{European symposium on algorithms}, 589--600. Springer.

\bibitem[{Kleinberg(2005)}]{kleinberg2005multiple}
Kleinberg, R. 2005.
\newblock A multiple-choice secretary algorithm with applications to online
  auctions.
\newblock In \emph{Proceedings of the sixteenth annual ACM-SIAM Symposium on
  Discrete Algorithms (SODA)}, 630--631. Citeseer.

\bibitem[{Korula and P{\'a}l(2009)}]{korula2009algorithms}
Korula, N.; and P{\'a}l, M. 2009.
\newblock Algorithms for secretary problems on graphs and hypergraphs.
\newblock In \emph{International Colloquium on Automata, Languages, and
  Programming (ICALP)}, 508--520. Springer.

\bibitem[{Krengel and Sucheston(1977)}]{krengel1977semiamarts}
Krengel, U.; and Sucheston, L. 1977.
\newblock Semiamarts and finite values.
\newblock \emph{Bulletin of the American Mathematical Society}, 83(4):
  745--747.

\bibitem[{Lachish(2014)}]{lachish2014log}
Lachish, O. 2014.
\newblock O (log log rank) competitive ratio for the matroid secretary problem.
\newblock In \emph{IEEE 55th Annual Symposium on Foundations of Computer
  Science (FOCS)}, 326--335.

\bibitem[{Li and Li(2020)}]{li2020fair}
Li, B.; and Li, Y. 2020.
\newblock Fair resource sharing and dorm assignment.
\newblock In \emph{International Conference on Autonomous Agents and MultiAgent
  Systems (AAMAS)}, 708--716.

\bibitem[{Mehta(2012)}]{mehta2012online}
Mehta, A. 2012.
\newblock Online Matching and Ad Allocation.
\newblock \emph{Theoretical Computer Science}, 8(4): 265--368.

\bibitem[{Motwani and Raghavan(1995)}]{motwani1995randomized}
Motwani, R.; and Raghavan, P. 1995.
\newblock \emph{Randomized algorithms}.
\newblock Cambridge university press.

\bibitem[{Preater(1993)}]{preater1993senior}
Preater, J. 1993.
\newblock The senior and junior secretaries problem.
\newblock \emph{Operations research letters}, 14(4): 231--235.

\bibitem[{Preater(1994)}]{preater1994multiple}
Preater, J. 1994.
\newblock On multiple choice secretary problems.
\newblock \emph{Mathematics of Operations Research}, 19(3): 597--602.

\bibitem[{Reiffenhauser(2019)}]{reiffenhauser2019optimal}
Reiffenhauser, R. 2019.
\newblock An optimal truthful mechanism for the online weighted bipartite
  matching problem.
\newblock In \emph{Proceedings of the Thirtieth Annual ACM-SIAM Symposium on
  Discrete Algorithms (SODA)}, 1982--1993.

\bibitem[{Soto(2013)}]{soto2013matroid}
Soto, J.~A. 2013.
\newblock Matroid secretary problem in the random-assignment model.
\newblock \emph{SIAM Journal on Computing}, 42(1): 178--211.

\bibitem[{Soto, Turkieltaub, and Verdugo(2021)}]{soto2021strong}
Soto, J.~A.; Turkieltaub, A.; and Verdugo, V. 2021.
\newblock Strong algorithms for the ordinal matroid secretary problem.
\newblock \emph{Mathematics of Operations Research}.

\end{thebibliography}
\newpage

\appendix
\section{Proof of Lemmas~\ref{lem:prob1-bm2} and~\ref{lem:prob2-bm2} for \textsc{BipartiteMatching2}} \label{appendix2}

\begin{proof}[Proof of Lemma ~\ref{lem:prob1-bm2}]
Recall that we number the vertices in the order that they arrive and we use the integer variable $v$ to denote both the number of a step and the vertex that arrives in that step. Recall that $M^v$ is the optimal matching computed in step $v$ and $e^v$ is the edge incident on $v$ in $M^v.$ The algorithm outputs matching $M.$ Denote the probability that $e^v$ can be added to the matching $M$ by $\P(v_{\texttt{success}})$ and the neighbor of $v$ via $e^v$ by $l(e^v).$  We start with a weak lower  bound on $\P(v_{\texttt{success}})$ and then improve it. 

An important aspect of \textsc{Alg2} that differentiates it from \textsc{Alg1} is that it doesn't exhaust the full capacity of any offline vertex by random matchings until every offline vertex is matched at least once. In  \textsc{Alg2}, $L_a$ is the set of offline vertices available for random matchings. When an offline vertex is matched, \textsc{Alg2} removes it from $L_a$. When $L_a$ becomes empty, i.e., when all offline vertices are matched at least once, it is reset to be the set of all vertices with available capacity. We make the following observation:

\textbf{Observation 1:} The event of resetting $L_a$ (which coincides with it becoming empty) happens only after step $\floor{\frac{n}{2}}.$ 

This is true because $|L_a|=\floor{\frac{n}{2}}$ at the start of the algorithm and at most one element is removed from it in a step. In any step $v \in \{k+1, \ldots, \floor{\frac{n}{2}}\},$ for vertex $l(e^v)$ to be available, it is sufficient that one of the following two disjoint events happens:
\begin{itemize}
\item \textit{Event $0_v$:} $l(e^v)$ is not the same as $l(e^u)$ for any $u < v.$ 

Since $l(e^v)$ is picked in a random matching at most once before being removed from $L_a$, it is available to be matched to $v$ in this event.

\item  \textit{Event $1_v$:} $l(e^v)$ is the same as $l(e^u)$ for exactly one $u < v,$ \emph{and} it was not picked in a random matching before step $u.$

In this event, once $l(e^v)$ is matched to $u,$ it is removed from the set $L_a.$ Therefore it is available in step $v.$  
\end{itemize}

See that these two events are disjoint, and therefore $\P(v_{\texttt{success}}) \geq \P(\textit{Event $0_v$}) + \P(\textit{Event $1_v$}).$ We now find lower bounds on the probabilities of these two events separately. 
Note that for finding lower bounds of $\P(v_{\texttt{success}})$, we consider the worst-case situation where $l(e^v)$ is part of matching $M^u$ for \emph{all} $u <v.$ This is because in the secretary model of online matching, the egde-weights are arbitrary and it is possible that $l(e^v)$ has high edge weights with all online vertices. However, we are saved by the uniformly random arrival order, which limits the probability of the event that $l(e^v)$ is matched to vertex $u$ in matching $M^u.$ Out of the $u$ participating online vertices in $M^u,$ $l(e^v)$ can be matched to at most $2$ vertices since $M^u$ is a capacitated bipartite matching with capacity $2$ for all offline vertices. The probability that vertex $u$ is one of these $2$ vertices matched to $l(e^v)$ is $2/u$ because of the the uniformly random arrival order of the $u$ participating vertices. Further, this event is independent of the event that $u' <u$ was matched to $l(e^v)$ in $M^{u'}.$ This independence holds because in step $u,$ we do not condition on the order of arrival of vertices $\{1, \ldots, u-1\}.$ Therefore,  $\P(\textit{Event $0_v$})$ satisfies the following:

\begin{align}
\P(\text{Event $0_v$}) &=  \prod_{u = k+1}^{v-1} \P[l(e^u)  \neq  l(e^v)] \geq   \prod_{u = k+1}^{v-1} \bigg(1-\frac{2}{u} \bigg) \nonumber\\
 &= \frac{k(k-1)}{(v-1)(v-2)} > \frac{k^2}{v^2}.  \label{eq:event0-bm1}
\end{align}
 The last inequality holds since for $k = \floor{\frac{n}{4}}$ and $ k < v \leq \floor{\frac{n}{2}},$ we have $\frac{k-1}{v-2} \geq \frac{k}{v}$ and $\frac{1}{v-1} > \frac{1}{v}.$ 
 
Now we obtain a lower bound on $\P(\text{Event $1_v$}).$ Denote the probability that vertex $l(e^v)$ is not picked in a random matching till step $v-1$ by $\P(\text{NR}_v).$ One simple lower bound on $\P(\text{NR}_v)$ is given by: 
\begin{align}
\P(\text{NR}_v) \geq \frac{n/2 -(v-1)}{n/2} = \frac{n-2v+2}{n}. \label{eq:no-rm-weak}
\end{align}

Inequality~\eqref{eq:no-rm-weak} holds because there are at most $v-1$ random matchings before step $v$ and each picks a different offline vertex out of $n/2$ options. We use Inequality~\eqref{eq:no-rm-weak} to compute a lower bound on $\P(\textit{Event $1_v$})$ as follows:
\begin{align}
&\!\!\!\!\P(\text{Event $1_v$})  \nonumber\\
&\!\!\!\!\!\!=\sum_{t = k+1}^{v-1} \!\! \bigg(  \P(\text{NR}_t)~\P[l(e^t)  =  l(e^v)]\! \prod_{\substack{ u = k+1\\u\neq t }}^{v-1} \! \P[l(e^u)  \neq  l(e^v)]\bigg),  \nonumber\\
&\!\!\!\!\!\!\geq  \sum_{t = k+1}^{v-1} \!\! \bigg(  \P(\text{NR}_t)~\P[l(e^t)  =  l(e^v)]\! \prod_{ u = k+1}^{v-1} \! \P[l(e^u)  \neq  l(e^v)]\bigg),  \nonumber\\
\intertext{Using Equations~\eqref{eq:no-rm-weak} and~\eqref{eq:event0-bm1} and the probability of the event $l(e^t)  =  l(e^v)$, we get:}
&\P(\text{Event $1_v$}) \geq  \sum_{t = k+1}^{v-1} \!\! \bigg(   \frac{n-2t+2}{n} \cdot \frac{2}{t} \cdot \frac{k^2}{v^2}\bigg),  \nonumber\\
\intertext{Interpreting the summation as a Riemann sum and using the smallest value of the function in each interval, we get the following integral:}
&\geq \frac{2k^2}{nv^2} \int_{t = k}^{v-1} \!\! \bigg( \frac{n-2t}{t-1} \bigg) dt, \nonumber\\
&\geq  \frac{2k^2}{nv^2} \left(n \log\left(\frac{v-2}{k-1}\right) - 2(v-2-k+1) \right),  \nonumber\\
\intertext{Capturing the lower order terms in $o(1),$ we get:}
&\P(\text{Event $1_v$}) \geq  \frac{2k^2}{v^2} \Big( \log\big(\frac{v}{k}\big) - \frac{2(v-k)}{n}- o(1) \Big).\label{eq:event1-bm1}
\end{align}
Combining Equations~\eqref{eq:event0-bm1} and~\eqref{eq:event1-bm1}, we get the following lower bound on $\P(v_{\texttt{success}}):$
\begin{align}
\P(v_{\texttt{success}}) &\geq \P(\textit{Event $0_v$}) + \P(\textit{Event $1_v$}), \nonumber \\
&= \frac{k^2}{v^2} +  \frac{2k^2}{v^2} \Big( \log\big(\frac{v}{k}\big) - \frac{2(v-k)}{n}- o(1) \Big), \nonumber \\
&= \frac{k^2}{v^2} \left(1 + 2\log\big(\frac{v}{k}\big) - \frac{4(v-k)}{n}- o(1) \right), \nonumber \\
\intertext{Setting $k = \floor{\frac{n}{4}}$ and capturing the lower order terms in $o(1)$, we get:}
\P(v_{\texttt{success}}) &\geq \frac{n^2}{16v^2} \left(2+2\log\big(\frac{4v}{n}\big) - \frac{4v}{n}- o(1) \right), \nonumber \\
\intertext{We use a  function linear in $\frac{v}{n}$ as a lower bound of the above function and get the following:}
\P(v_{\texttt{success}}) &\geq \frac{2(n-2v)}{n} - o(1). \label{eq:weak-p-success}
\end{align}
The above inequality holds in the interval $k < v \leq n$ and is not obvious unless plotted and observed. For intuition, note that the two expressions are equal at $v = \frac{n}{4}$ and the former expression decays slower in $v$.

 We use the lower bound on $\P(v_{\texttt{success}})$ in Equation~\eqref{eq:weak-p-success} to find a better lower bound on $\P(\text{NR}_v).$ Note that there is a random matching in step $u>k$ with probability $1-\P(v_{\texttt{success}}).$ Therefore we have the following:
  \begin{align}
 &\P(\text{NR}_v) \nonumber\\ 
  &=\prod_{u = 1}^k \Big(1 - \frac{1}{\frac{n}{2}-(u-1)} \Big) \prod_{u=k+1}^{v-1} \Big(1 -\frac{1-\P(v_{\texttt{success}})}{\frac{n}{2}-(u-1)} \Big), \nonumber \\
  &\geq \prod_{u = 1}^k  \frac{\frac{n}{2}-u}{\frac{n}{2}-u+1}   \prod_{u=k+1}^{v-1} \frac{ \frac{n}{2}-u+1 -1+\P(v_{\texttt{success}})}{\frac{n}{2}-u+1}, \nonumber \\
  &= \prod_{u = 1}^k \frac{\frac{n}{2}-u}{\frac{n}{2}-u+1} \prod_{u=k+1}^{v-1} \frac{ (\frac{n}{2}-u) \Big[1 + \frac{\P(v_{\texttt{success}})}{(\frac{n}{2}-u)}\Big]}{\frac{n}{2}-u+1}, \nonumber \\
   &= \prod_{u = 1}^{v-1} \frac{\frac{n}{2}-u}{\frac{n}{2}-u+1}  \prod_{u=k+1}^{v-1} \left[1 + \frac{\P(v_{\texttt{success}})}{(\frac{n}{2}-u)}\right], \nonumber \\
   \intertext{Using the lower bound on $\P(v_{\texttt{success}})$ from Equation~\eqref{eq:weak-p-success},}
   &\geq \frac{\frac{n}{2}- (v-1)}{\frac{n}{2}} \prod_{u=k+1}^{v-1}  \left[1 +  \frac{4}{n}  -\frac{o(1)}{n} \right], \nonumber \\
&\geq \frac{n- 2v}{n}   \left[1 + \sum_{u=k+1}^{v-1} \frac{4}{n}\right], \nonumber \\
   &= \frac{n- 2v}{n} \left(1 + \frac{4}{n} \left( v-k-2 \right) -o(1) \right).\nonumber \\
   \intertext{Plugging in the value of $k$,}
   &\P(\text{NR}_v) = \frac{4v(n- 2v)}{n^2} - o\left(\frac{1}{n}\right). \label{eq:nr-stronger}
  \end{align} 
  Using the improved lower bound on $\P(\text{NR}_v)$ in Equation~\eqref{eq:nr-stronger}, we get a better lower bound on $\P(\text{Event $1_v$})$ than in Equation~\eqref{eq:event1-bm1}.
  \begin{align}
&\P(\text{Event $1_v$})  \nonumber\\
&\!\!=  \!\!\sum_{t = k+1}^{v-1} \!\! \bigg(  \P(\text{NR}_t)~\P[l(e^t)  =  l(e^v)]\! \prod_{\substack{ u = k+1\\u\neq t }}^{v-1} \! \P[l(e^u)  \neq  l(e^v)]\bigg),  \nonumber\\
&\!\!\geq \!\! \sum_{t = k+1}^{v-1} \!\! \bigg(  \P(\text{NR}_t)~\P[l(e^t)  =  l(e^v)]\! \prod_{ u = k+1}^{v-1} \! \P[l(e^u)  \neq  l(e^v)]\bigg),  \nonumber
\end{align}
\begin{align}
&\!\!\geq \!\! \sum_{t = k+1}^{v-1} \!\! \bigg( \left( \frac{4t(n- 2t)}{n^2} - o\left(\frac{1}{n}\right) \right)\cdot \frac{2}{t} \cdot \frac{k^2}{v^2}\bigg),  \nonumber \\
\intertext{Converting the Riemann summation into integration:}
&= \frac{8k^2}{n^2v^2} \int_{t = k}^{v-1} (n-2t) - o(n) ~dt , \nonumber\\
&= \frac{8k^2}{n^2v^2} \left( (v-k)n -(v^2 - k^2) - o(n^2)\right),  \nonumber\\
&=  \frac{8k^2}{n^2v^2} \left( (v-k)(n -v - k) - o(n^2)\right),  \nonumber\\
\intertext{Plugging in the value of $k$, we get:}
&\P(\text{Event $1_v$})  \geq  \frac{16nv-3n^2-16v^2 - o(n^2)}{32v^2}. \nonumber \\
\intertext{Recall Equation~\eqref{eq:event0-bm1},}
&\P(\text{Event $0_v$}) \geq \frac{k^2}{v^2} = \frac{n^2}{16v^2} -o(1). \nonumber \\
\intertext{Therefore,}
&\P(v_{\texttt{success}}) \geq \P(\textit{Event $0_v$}) + \P(\textit{Event $1_v$}) \nonumber \\
&\geq \frac{n^2}{16v^2} +  \frac{16nv-3n^2-16v^2}{32v^2} - o(1),\nonumber \\
&= \frac{n-v}{2v} -\frac{n^2}{32v^2}  - o(1).\nonumber
\end{align}
This completes the proof.
\end{proof}
Now we use the above result to give a proof of Lemma~\ref{lem:prob2-bm2}.
\begin{proof}[Proof of Lemma~\ref{lem:prob2-bm2}]
For steps $v > \floor{\frac{n}{2}},$ offline vertex $l(e^v)$ is available in step $v$ if the following happens:
\begin{itemize}
\item $l(e^v)$ was not matched to full capacity in the first $\floor{\frac{n}{2}}$ steps, \textit{and}
\item $l(e^v) \neq l(e^u)$ for any step $u$ such that $\frac{n}{2} < u < v,$ \textit{and}
\item $l(e^v)$ is not picked in a random matching in any step $u$ for $\frac{n}{2} < u < v.$ 
\end{itemize}
The first event happens with probability $\P(\floor{\frac{n}{2}}_\texttt{success}) \geq  \frac{3}{8} - o(1).$ The second event, independent from the first, happens with probability given by:
\begin{align}
&\prod_{u = \frac{n}{2}+1}^{v-1} \P[l(e^u)  \neq  l(e^v)] \geq   \prod_{u = \frac{n}{2}+1}^{v-1} \bigg(1-\frac{2}{u} \bigg), \nonumber\\
 &= \frac{\frac{n}{2}(\frac{n}{2}-1)}{(v-1)(v-2)} > \frac{n^2}{v^2}.  
\end{align}
The third event happens with probability at least $\frac{2(n-v)}{n}.$ This is because an offline vertex that was already matched earlier is picked in a random matching only after $L_a$ is reset. Let $L_a$ be reset in step $t.$ By Observation 1, $t \geq \floor{\frac{n}{2}}$. Then the probability of the third event is at least $\frac{(n-t)-(v-t-1)}{n-t},$ which is minimized for $t = \floor{\frac{n}{2}}$. The statement of the lemma follows from the product of the probabilities of all three events: $\frac{3}{8} - o(1), \frac{n^2}{v^2},$ and $\frac{2(n-v)}{n}$ .
\end{proof}

\end{document}